\documentclass[11pt, rqno]{article}

\usepackage{titlesec}
\titleformat*{\section}{\bf\Large} 
\titleformat*{\subsection}{\bf\large} 

\usepackage{amsmath}
\usepackage{amssymb}
\usepackage{amsfonts}
\usepackage{setspace}
\usepackage{epsfig}
\usepackage{amsthm}
\usepackage{graphicx}
\usepackage{color}

\newcommand{\be}{\begin{equation}}
\newcommand{\ee}{\end{equation}}

\newcommand{\vs}{\vspace{0.2cm}}

\newtheorem{Theorem}{Theorem}[section]

\newtheorem{Lemma}[Theorem]{Lemma}
\newtheorem{Corollary}[Theorem]{Corollary}

\newcommand{\dist}{{\rm dist}}

\newcommand{\diss}{\displaystyle}

\newcommand{\length}{{\rm length}}

\title{On static solutions of the Einstein - Scalar Field equations} 
\author{Mart\'in Reiris\footnote{Email: mareithu@gmail.com. This work was finished while the author was a postdoc at the Max Planck Institute for Gravitational Physics. Golm, Germany}}

\date{}

\begin{document}
\maketitle
\vspace{-1cm}
\begin{abstract}
{In this article we study self-gravitating static solutions of the Einstein-ScalarField system in arbitrary dimensions. We discuss the existence and the non-existence of geodesically complete solutions depending on the form of the scalar field potential $V(\phi)$, and provide full global geometric estimates when the solutions exist. Our main results are summarised as follows. For the Klein-Gordon field, namely when $V(\phi)=m^{2}|\phi|^{2}$, it is proved that geodesically complete solutions have Ricci-flat spatial metric, have constant lapse and are vacuum, (that is $\phi$  is constant and equal to zero if $m\neq 0$). In particular, when the spatial dimension is three, the only such solutions are either Minkowski or a quotient thereof (no nontrivial solutions exist). When $V(\phi)=m^{2}|\phi|^{2}+2\Lambda$, that is, when a vacuum energy or a cosmological constant is included, it is proved that no geodesically complete solution exists when $\Lambda>0$, whereas when $\Lambda<0$ it is proved that no non-vacuum geodesically complete solution exists unless $m^{2}<-2\Lambda/(n-1)$, ($n$ is the spatial dimension) and the spatial manifold is non-compact. The proofs are based on techniques in comparison geometry {\it \'a la Backry-Emery.}} 
\end{abstract}

\maketitle
\section{Introduction}

A classical result in General Relativity due to Lichnerowicz \cite{Lichnerowicz} (with previous work by Einstein and Einstein-Pauli, see \cite{10.2307/1968759}) asserts that there are no nontrivial asymptotically flat stationary solitons\footnote{By soliton we understand a `particlelike' regular and localised solution.} of the vacuum Einstein equations. That is, the only such solution is the Minkowski spacetime. This result is generalisable to include matter like the electromagnetic field or the Klein-Gordon field, but not to any type of non-exotic matter. Indeed Bartnik and McKennon \cite{Bartnik:1988am}, (and later rigorously Smoller-Wasserman-Yau-McLeod \cite{ref1}), found a remarkable soliton for the Einstein-YangMills system with gauge group SU(2), despite of the fact that there are no nontrivial solitons for the SU(2) Yang-Mills theory alone. The discovery of this soliton opened an interesting new window of research continuing until today. 

On the other hand, a fundamental extension of Lichnerowicz's theorem in vacuum due to Anderson \cite{MR1806984} asserts that the only geodesically complete solution of the strictly stationary Einstein equations is Minkowski or a quotient thereof. Thus, for strictly stationary solutions, no material sources implies no gravity and this is true in any possible geodesically complete scenario. Anderson's theorem uses fundamentally a well known conformal transformation that presents the stationary equations as a harmonic map into hyperbolic two-space (see \cite{MR2003646}). Thinking in possible extensions, the drawback of this presentation is that it is seldom possible when material fields are considered. In addition Anderson's proof relies heavily on special properties of the Cheeger-Gromov theory of convergence and collapse of Riemannian manifolds in dimension three that do not hold in higher dimensions. 

Recently however, Cortier and Minerve \cite{Cortier:2014voa} were able to reprove Anderson's result (under extra assumptions) using just standard techniques in comparison geometry. Yet their result relies again on the harmonic-map representation earlier mentioned. On the other hand, Case \cite{MR2741248} also motivated by Anderson's theorem, was able to apply techniques in comparison geometry \'a la Backry-Emery to prove closely related rigidity results for quasi-Einstein metrics bearing much in common with the static Einstein equations and with certain types of Einstein-ScalarField equations. Case's technique is considerably flexible and generalisable to higher dimensions. We note also that there is recent interesting work by Galloway-Woolgar \cite{Galloway:2013cea} and Woolgar-Wylie \cite{Woolgar:2015wca} on the so called Backry-Emery spacetimes familiar in some way to Case's paper.  

Motivated by these developments, in this article we import techniques in comparison geometry \'a la Backry-Emery to study in depth the existence or non-existence of {\it geodesically complete} static solutions of the Einstein-ScalarField equations for different types of the scalar field potential $V(\phi)$. We do not make any dimensional, global, or even any asymptotic assumption like asymptotic flatness, and in this sense several of the conclusions of this paper are the most general they can be.

As a main application, with a marked physical interest, we discuss thoroughly the ubiquitous Klein-Gordon field in the presence or not of a cosmological constant. Let us describe the conclusions in full detail. For the setup we refer the reader to the next section. We divide the discussion according to the type of potential $V(\phi)$. When $V(\phi)$ is just the standard Klein-Gordon potential, i.e. $V(\phi)=m^{2}|\phi|^{2}$, it is proved that geodesically complete solutions have Ricci-flat spatial metric, have constant lapse $N$, and are vacuum, that is $\phi=\phi_{0}$ with $\phi_{0}=0$ if $m\neq 0$, ($\S$ Theorem \ref{EKGLZ}). Therefore, if the spatial dimension is three, the only such solutions are either Minkowski or a quotient thereof. When $V(\phi)=m^{2}|\phi|^{2}+2\Lambda$, that is, including a vacuum energy or a cosmological constant, we prove that no geodesically complete solution exists when $\Lambda>0$, whereas when $\Lambda<0$ it is proved that no non-vacuum geodesically complete solution exists unless $m^{2}<-2\Lambda/(n-1)$ and unless the manifold is non-compact, ($\S$ Theorem \ref{NCOMP} and Theorem \ref{LPHIEST}). Moreover, in this case, we provide the pointwise estimate $|\nabla \phi|^{2}+m^{2}|\phi|^{2}\leq -68\Lambda$ for the energy density ($\S$ Theorem \ref{LPHIEST}), the pointwise estimate $|\nabla N|/N\leq 64\sqrt{-\Lambda}$ for the gradient of the lapse ($\S$ Theorem \ref{EKGLEDEST}), and, when the spatial dimension is three, we prove a general pointwise bound on the curvature in terms of $|\Lambda|$ ($\S$ Theorem \ref{RCEST}). In this last case, vacuum solutions (i.e. $\phi=0$) other than AdS were shown to exist by Anderson in \cite{OSCCEM} though it is not know if non-vacuum solutions exist. The pointwise estimates that we obtain seem to be new even in the vacuum case. These estimates could be useful in theories that study spaces asymptotic to AdS, with or without a scalar field.


One could easily guess that, by following either the line of argument used in this paper for the Klein-Gordon field or by taking other original paths, many other new applications of the techniques could be found. As an instance of that, in Section \ref{RSF} we enumerate briefly a series of conclusions (mostly `no-go' theorems) that one can easily reach for the Einstein-(real)ScalarField system, for several different types of potentials $V(\phi)$ including the (real)-Sine-Gordon and the (real)-Higgs potentials. 


We elaborate now on the method of proof. To simplify the whole presentation we opted to base our results on a main technical Lemma \ref{MAIN}, which was adapted from \cite{MR2741248}, and from which all the main theorems are directly obtained. 
%
%
The applications, which are elaborated in Section \ref{APPL}, are deduced from Lemma \ref{MAIN} in conjunction with a simple observation that is worth to mention here. The key observation is that, using the static equations and the B\"ochner-type of equation (\ref{BOC}), one can obtain expressions for the $f$-Laplacian $\Delta_{f}$, (with $f=-\ln N$, see next section), of $\psi=|\nabla \phi|^{2}$ and of $\psi=|\nabla \ln N|^{2}$ of the form
\be\label{firsteq}
\Delta_{f}\psi\geq b\psi+c\psi^{2},
\ee
with $b\leq 0$ and $c>0$. The Lemma \ref{MAIN} is then used to provide fundamental pointwise estimates for a $\psi$ satisfying (\ref{firsteq}) whenever the $f$-Ricci tensor $Ric^{1}_{f}$ (see next section) is bounded below. As it happens that in the main applications $Ric_{f}$ is bounded below, this provides the fundamental gradient estimates for $\phi$ and $\ln N$ from which all the conclusions of this paper follow. 

The organisation of the article is as follows. In Section \ref{THSE} we recall the main static equations of the Einstein-ScalarField system, together with the notation and the terminology. Subsection \ref{THSEA} explains the type of manifolds used during the paper, and that have to be read with care to avoid confusion. Section \ref{THTL} is the technical section where the main Lemma \ref{MAIN} is stated and proved. The applications are discussed in Section \ref{APPL} which is divided in three subsections: Subs. \ref{KGSUB} discusses the Klein-Gordon case, Subs. \ref{LKGSUB} discusses the Klein-Gordon case in the presence of a cosmological constant and Subs. \ref{RSF} elaborates on applications to real-Scalar Fields. 

\section{The static equations}\label{THSE}  
We give below the static equations of the Einstein-ComplexScalarField system in spacetime-dimension $n+1$, ($n\geq 2$). We use the following notation: (i) $\phi$ is the complex scalar field and $\bar{\phi}$ the complex conjugate (ii) $\phi_{R}$ is the real part of $\phi$ and $\phi_{I}$ the imaginary part (iii) $|\phi|$ is the norm of $\phi$ and $|\nabla \phi|$ is the norm of $\nabla \phi$ (i.e. $|\nabla \phi|^{2}=\langle \nabla \phi,\nabla \bar{\phi}\rangle$). The potentials that we will consider are of the form $V(\phi_{R},\phi_{I})$. We will use the shorthand $V(\phi)$. The spacetime metric is assumed to split as
${\bf g}=-N^{2}dt^{2}+g$, and the metric $g$, as well as the lapse $N>0$, live in a $n$-dimensional manifold $\Sigma$. The relevant data is thus $(\Sigma;N,g;\phi)$.

The static Einstein-(Complex)ScalarField equations are,
\begin{align}
\label{SESF1} & Ric +\nabla\nabla f-\nabla f\nabla f=\nabla \phi\circ \nabla \bar{\phi} +\frac{V(\phi)}{n-1} g,\\
\label{SESF3} & \Delta f -\langle \nabla f,\nabla f\rangle=\frac{V(\phi)}{n-1},\\
\label{SESF4} & \Delta \phi -\langle\nabla f,\nabla\phi\rangle=\frac{\partial V(\phi)}{2},
\end{align}
where $f=-\ln N$, $\nabla \phi\circ \nabla \bar{\phi}=(\nabla \phi \nabla \bar{\phi}+\nabla \bar{\phi} \nabla \phi)/2=\nabla \phi_{R}\nabla \phi_{R}+\nabla\phi_{I}\nabla \phi_{I}$, and $\partial V$ is
\be
\partial V=\frac{\partial V}{\partial \phi_{R}}+i\frac{\partial V}{\partial \phi_{I}}
\ee
These equations imply, in turn, the following expression for the scalar curvature $R$ (or energy density),
\be
\label{SESF2} R=|\nabla \phi|^{2}+V(\phi).
\ee

The system (\ref{SESF1})-(\ref{SESF3}) arises as the {\it static} Euler-Langrange equations of the $n+1$-dimensional (spacetime) action\footnote{Inserting positive constants in front of $R$, or $\nabla_{\mu}\phi \nabla^{\mu} \phi$ does not change the analysis of this article.} 
\be
\mathcal{S}({\bf g},\phi)=\int \bigg[\, R_{\bf g}-\nabla_{\mu}\phi\nabla^{\mu}\bar{\phi}-V(\phi)\bigg]\, d {\rm v}_{\bf g}
\ee
or, also, as the Euler-Lagrange equations of the $n$-dimensional (spatial) action
\be
\mathcal{S}(f,g,\phi)=\int \bigg[R-|\nabla \phi|^{2}-V(\phi)\bigg]e^{-f} d{\rm v}
\ee
A few times below we will use, following \cite{MR2577473}, the notation
\be
Ric^{1}_{f}:=Ric+\nabla\nabla f-\nabla f\nabla f
\ee
and
\be
\Delta_{f}\psi=\Delta\psi-\langle \nabla f,\nabla \psi\rangle=0
\ee 
\subsection{Manifolds}\label{THSEA}
Without any explicit specification, a `{\it manifold} $\Sigma$' is allowed to have boundary or to be boundaryless, and to be compact or non-compact. Whatever the case, $(\Sigma,g)$ is assumed metrically complete with respect to the standard metric
\be
\dist(p,q)=\inf\{\length(\gamma_{pq}): \gamma_{pq}\in \mathcal{C}_{pq}\}
\ee
where $\mathcal{C}_{pq}$ is the set of smooth curves joining $p$ to $q$. Hence, if $\Sigma$ is boundaryless then $(\Sigma,g)$ is {\it geodesically complete} by Hopf-Rinow. On the other hand if $\Sigma$ has boundary then $(\Sigma,g)$ is {\it geodesically incomplete} as geodesics can terminate at the boundary. Henceforth, when we say `{\it $(\Sigma,g)$ is geodesically complete}', we are saying implicitly that $\Sigma$ is boundaryless, no matter if $\Sigma$ is compact or not. 

These conventions have to be kept in mind to prevent confusion. For example, the de-Sitter metric
\be
g=\bigg(\frac{1}{1-\Lambda r^{2}/3}\bigg)dr^{2}+r^{2}d\Omega^{2}_{n-2},\ N=\sqrt{1-\frac{\Lambda r^{2}}{3}}
\ee
is a solution of the static Einstein equations with a positive cosmological constant, although we will show later that there is no such solution which is geodesically complete. The point here is that the de-Sitter solution is defined on a manifold with boundary (the cosmological horizon), hence geodesically incomplete.

As was detailed earlier, several of the main conclusions of this article are about {\it geodesically complete static spacetimes}. Namely, geodesics of any spacetime character have infinite parametric time. If, as we are assuming, such spacetime splits as $\mathbb{R}\times \Sigma$ with ${\bf g}=-N^{2}dt^{2}+g$ then the slice $(\Sigma; g)$ is also geodesically complete because $\Sigma$ is totally geodesic inside ${\bf M}$. Hence we do not loose generality, (rather we gain), if the statements later refer only to the geodesic completeness of the $(\Sigma; g)$. This is more convenient as the techniques that we use for the different proofs deal only with $(\Sigma; g)$ and not with $({\bf M}; {\bf g})$. 

A non-existence result of geodesically complete solutions is important because it says that any inextensible solution (necessarily geodesically incomplete) has always, roughly speaking, either a horizon or a singularity.
\section{The technical lemmas}\label{THTL}
In this section we state and prove Lemma \ref{MAIN} which is the main technical lemma to be used in applications. We start recalling (without proof) Theorem A.1 from \cite{MR2577473}. To get exactly the expression (\ref{FLAPEST}) from the statement of Theorem A.1 simple replace $m^{n+1}_{H}$ for its equivalent in eq. (3.8) of \cite{MR2577473}. Lemma \ref{MAIN} is discussed afterwards. 

In this Theorem and also below $d_{p}$ is equal to either 
\be
d_{p}=\dist(p,\partial \Sigma),
\ee
if $\partial \Sigma\neq \emptyset$ or
\be
d_{p}=\sup\{\dist(p,x):x\in \Sigma\}
\ee
if $\partial \Sigma = \emptyset$. In particular if $\Sigma$ is non-compact and boundaryless then $d_{p}=\infty$, \footnote{A technical remark is here necessary. For $p\in \Sigma\setminus \partial \Sigma$ and $r_p<d_{p}$ the metric ball $B(p,r_p):=\{q\in \Sigma: \dist(p,q)<r\}$ has the following property: for every $q$ in $B(p,r_p)$, there is at least a length minimising segment joining $p$ to $q$ and entirely inside $B(p,r_p)$. Thus, inside $B(p,r_p)$, the distance function $r(q)=\dist(p,q)$ can be used as any geodesic distance function. These properties may not hold if $r_p>d_{p}$ and this explain why we need the condition $r_p<d_p$ in Theorem \ref{LWW}.}

\begin{Theorem}{\rm (\cite{MR2577473})}\label{LWW} Let $(\Sigma,g)$ be an $n$-dimensional Riemannian manifold. Suppose that
\be
Ric+\nabla\nabla f-\nabla f\nabla f\geq (nH)g
\ee
for some function $f$ and real number $H$. Let $p$ be a point in $\Sigma\setminus \partial \Sigma$ and let $r$ be the distance function to $p$, i.e. $r(x)=\dist(x,p)$. Then, at any $x$ such that $r(x)<d_{p}$ we have
\be\label{FLAPEST}
\Delta_{f}r\leq 
\left\{
\begin{array}{ll}
\frac{\diss n\sqrt{H}}{\diss \tan (\sqrt{H}r)} & {\rm if\ } H>0,\vs\\
\frac{\diss n}{\diss r} & {\rm if\ } H=0,\vs\\
\frac{\diss n\sqrt{|H|}}{\diss \tanh (\sqrt{|H|}r)} & {\rm if\ } H<0.
\end{array}
\right.
\ee
in the barrier sense\footnote{This is an important property as it allows us to make analysis as if $r$ were a smooth function. The reader can consult this notion in \cite{MR2243772}.}. 
\end{Theorem}
Of course we could have $\partial \Sigma=\emptyset$ in which case $\Sigma\setminus \partial \Sigma=\Sigma$. As seen in \cite{MR2577473}, this Theorem implies the following generalised Myers's estimate: if $H>0$, then for any point $p$ we have $d_{p}\leq \pi/\sqrt{H}$. In particular if $\Sigma$ is non-compact then $\partial \Sigma\neq \emptyset$. We will use this property later.

The following is the main technical Lemma to be used and that is adapted from an estimate due to Case \cite{MR2741248}.
\begin{Lemma}\label{MAIN} Let $(\Sigma;g)$, $f$, $H$ and $p$ be as in Theorem \ref{LWW}. Let $\psi$ be a real non-negative function such that
\be\label{FUNDEQ}
\Delta_{f}\psi\geq b\psi+c\psi^{2},
\ee
with $b\leq 0$ and $c>0$. Then,
\be\label{PSITIM}
\psi(p)\leq 
\left\{
\begin{array}{ll}
\frac{\diss 1}{\diss c}\bigg[\frac{\diss 4n+24}{\diss d_{p}^{2}}-b\bigg] & {\rm if\ } H\geq 0,\\
\frac{\diss 1}{\diss c}\bigg[\diss \frac{4n\sqrt{|H|}}{\diss d_{p}\tanh (\sqrt{|H|} d_{p})}+\frac{\diss 24}{\diss d_{p}^{2}}-b\bigg] & {\rm if\ }
H<0
\end{array}
\right.
\ee
\end{Lemma}
\begin{proof}For any function $\chi$ the following general formula holds
\be
\Delta_{f}(\chi\psi)=\psi\Delta_{f} \chi+2\langle \nabla\chi,\nabla \psi\rangle +\chi\Delta_{f}\psi
\ee
Thus, if $\chi\geq 0$ and if $\chi\psi$ has a local maximum at $q$, then we have
\begin{align}\label{PRELCALC}
0 & \geq \bigg[\Delta_{f}(\chi\psi)\bigg]\bigg|_{q} \\ 
& \geq  \bigg[\psi \Delta_{f}\chi - 2\frac{|\nabla \chi|^{2}}{\chi}\psi + b\chi\psi+c\chi\psi^{2}\bigg]\bigg|_{q}
\end{align}
where to obtain the second inequality we used (\ref{FUNDEQ}). To simplify notation let $r=r(x)=\dist(x,p)$. Let $r_{p}$ be a positive number less than $d_{p}$. On the ball $B(p,r_p)$ let the function $\chi(x)$ be $\chi(x)=(r_p^{2}-r^{2}(x))^{2}$. Let $q$ be a point in the closure of $B(p,r_p)$ where the maximum of $\chi\psi$ is achieved. As $(\chi\psi)(q)\geq(\chi\psi)(p) = r_{p}^{4}\psi(p)$ we deduce that if $\psi(q) = 0$ then $\psi(p)=0$. In this case (\ref{PSITIM}) follows. So let us assume that $\psi(q)>0$ and hence that $q\in B(p,r_{p})$. By (\ref{PRELCALC}) we have
\begin{align}
\label{PREVIOUS} cr_p^{4}\psi(p) & \leq c(\chi\psi)(q) \leq \bigg[2\frac{|\nabla\chi|^{2}}{\chi}-\Delta_{f}\chi - b\chi \bigg]\bigg|_{q}\\
\label{PREVIOUSII} & \leq \bigg[4(r_p^{2}-r^{2})r\Delta_{f}r+4r_p^{2}+20r^{2}-b r_{p}^{4} \bigg]\bigg|_{q}
\end{align}
But if $Ric_{f}^{\alpha}\geq nH g$ then $\Delta_{f} r$ can be estimated from (\ref{FLAPEST}). Use this estimation in (\ref{PREVIOUSII}), divide by $cr^{4}_{p}$, and take the limit $r_{p}\rightarrow d_{p}$ to obtain (\ref{PSITIM}) by simple bounds.
\end{proof}
\begin{Corollary}\label{CCOR} Assume the hypothesis of Lemma \ref{MAIN} and that $(\Sigma;g)$ is non-compact and geodesically complete. Then,
\be
\psi(p)\leq -\frac{b}{c}
\ee
at any $p\in \Sigma$, regardless of the sign of $H$. 
\end{Corollary}
\begin{proof} If $(\Sigma;g)$ is non-compact and geodesically complete then $d_{p}=\infty$ and the result follows from (\ref{PSITIM}).
\end{proof}
\section{Applications}\label{APPL}
\subsection{Klein-Gordon}\label{KGSUB}
In this section we study the Klein-Gordon potential $V(|\phi|)=m^{2}|\phi|^{2}$. The mass is allowed to be zero in which case $V=0$. The theorem that follows is perhaps the simplest and most elegant application of the estimates of the previous section.
\begin{Theorem}\label{EKGLZ} Let $(\Sigma; N, g, \phi)$ be a geodesically complete solution of the $n$-dimensional static Einstein-KleinGordon equations. Then, $Ric=0$, $N=N_{0}$ and $\phi=\phi_{0}$, with $\phi_{0}=0$ if $m\neq 0$. In particular if $n=3$ then $(\Sigma;g)$ is covered by the Euclidean three-space. 
\end{Theorem}
The main B\"ochner type of formula that we are going to use is
\begin{align}\label{BOC}
\frac{1}{2}\Delta_{f}|\nabla \chi|^{2} = & |\nabla\nabla \chi|^{2}+\langle \nabla \chi,\nabla(\Delta_{f}\chi)\rangle\\
& +Ric^{1}_{f}(\nabla\chi,\nabla\chi)+|\langle \nabla\chi,\nabla f\rangle|^{2}
\end{align}
which is valid for any real function $\chi$, \cite{MR2577473}.
\begin{proof} During the proof we make $f=-\ln N$. To start note that if $\Sigma$ is compact and $m\neq 0$ then $\phi=0$ by integrating (\ref{SESF3}) against $N=e^{-f}$. But if $\phi=0$ then $f$ is constant by integrating (\ref{SESF3}) against $N^{2}=e^{-2f}$. Thus $Ric=0$ by (\ref{SESF3}) as claimed. Identical conclusion is reached if $m=0$ by integrating (\ref{SESF3}) agains $N^{2}=e^{-2f}$ and (\ref{SESF4}) against $\phi$.

Assume then from now on that $\Sigma$ is non-compact. Recall that we use the notation $\phi=\phi_{R}+i\phi_{I}$. From (\ref{SESF4}) we obtain 
\be
\Delta_{f}\chi=m^{2}\chi
\ee
for $\chi$ equal to either $\phi_{R}$ or $\phi_{I}$. Use then these two equations to evaluate (\ref{BOC}) with $\chi=\phi_{R}$ and with $\chi=\phi_{I}$. Add up the results and get (after discarding a few positive terms)
\be
\frac{1}{2}\Delta_{f}(|\nabla \phi_{R}|^{2}+|\nabla \phi_{I}|^{2})\geq |\nabla \phi_{R}|^{4}+|\nabla \phi_{I}|^{4}
\ee
Use now $|\nabla \phi|^{2}=|\nabla \phi_{R}|^{2}+|\nabla \phi_{I}|^{2}$ and the inequality $(x^{4}+y^{4})\geq (x^{2}+y^{2})^{2}/2$ to arrive at
\be
\Delta_{f}|\nabla \phi|^{2}\geq |\nabla \phi|^{4}
\ee    
It follows then from Corollary \ref{CCOR} that $\nabla\phi=0$. Hence $\phi=\phi_{0}$, and $\phi_{0}=0$ if $m\neq 0$ from (\ref{SESF4}).

We prove now that the lapse must be constant. From what was proved before we have $Ric^{1}_{f}=0$ and $\Delta_{f}f=0$. Use then (\ref{BOC}) with $\chi=f$ to get
\be
\frac{1}{2}\Delta_{f}|\nabla f|^{2}\geq |\nabla f|^{4}
\ee
Thus, $\nabla f=0$ from Corollary \ref{CCOR} and hence $N=N_{0}$. 

If $f=f_{0}$ then $Ric=0$ from $Ric^{1}_{f}=0$.
\end{proof}
\subsection{$\Lambda$-Klein-Gordon}\label{LKGSUB}
In this section we investigate geodesically complete solutions of the static Einstein-ScalarField equations with potentials of the form $V(\phi)=m^{2}|\phi|^{2}+2\Lambda$. 

The case $\Lambda=0$ was the one considered in the previous section, therefore we consider below only the cases $\Lambda>0$ and $\Lambda<0$. 

$\Lambda > 0$: In this case it is easy to see that there are no geodesically complete solutions at all. Indeed, if $\Sigma$ is compact a contradiction is obtained by integrating (\ref{SESF3}) against $N=e^{-f}$. On the other hand if $\Sigma$ is non-compact, then $\Sigma$ must have boundary because $Ric^{1}_{f}\geq (2\Lambda/(n-1))g$, as we already commented after the statement of Theorem \ref{LWW}. This thus contradicts the assumption that $(\Sigma,g)$ is geodesically complete.

$\Lambda<0$: As explained in the introduction, there are geodesically complete solutions in this case, therefore the best one can do is to understand the local and global geometry. Our first results shows that geodesically complete solutions with $\Sigma$ compact do not exist. Our second result uses this information to provide complete estimates on the scalar field $\phi$. 
\begin{Theorem}\label{NCOMP} Let $(\Sigma;N,g,\phi)$ be a geodesically complete solution of the static Einstein-ScalarField equations with potential $V(\phi)=m^{2}|\phi|^{2}+2\Lambda$, where $\Lambda<0$. Then $\Sigma$ is non-compact. 
\end{Theorem}
\begin{proof} During the proof we use $f=-\ln N$. Assume that $\Sigma$ is compact. Then observe that as (\ref{SESF4}) is equivalent to 
\be
\nabla(N\nabla \phi)=m^{2}N\phi
\ee
we can multiply this equation by $\bar{\phi}$ an integrate over $\Sigma$ to obtain 
\be
0=\int_{\Sigma}N(|\nabla \phi|^{2}+m^{2}|\phi|^{2})\, d {\rm v}
\ee
This implies $\phi=\phi_{0}$ with $\phi_{0}=0$ if $m\neq 0$. Using this information then note that (\ref{SESF3}) is equivalent to $\Delta N=(2\Lambda/(n-1))N$. Integrating this over $\Sigma$ we deduce $\Lambda=0$, and thus a contradiction.
\end{proof}
\begin{Theorem}\label{LPHIEST} Let $(\Sigma;g,N,\phi)$ be a geodesically complete solution of the static Einstein - Scalar Field equations with potential 
$V(\phi)=m^{2}|\phi|^{2}+2\Lambda$, where $\Lambda<0$. Then the following holds: 
\begin{enumerate}
\item[(i)] if $m^{2}\geq -2\Lambda/(n-1)$ then $\phi$ is identically zero, and,
\item[(ii)] if $m^{2}< -2\Lambda/(n-1)$ then,
\be\label{PHIEST}
|\nabla \phi|^{2}\leq \frac{-4\Lambda}{(n-1)},\qquad m^{2}|\phi|^{2}\leq -64\Lambda.
\ee
In particular $R=|\nabla \phi|^{2}+m^{2}|\phi|^{2}+2\Lambda\leq -66\Lambda$, by a coarse estimation.   
\end{enumerate}
\end{Theorem}
\begin{proof} During the proof we use $f=-\ln N$. Use (\ref{BOC}) with $\chi=\phi_{R}$ and with $\chi=\phi_{I}$ and add up the results to obtain (after discarding a few positive terms)
\be
\Delta |\nabla \phi|^{2}\geq 2(m^{2}+\frac{2\Lambda}{n-1})|\nabla \phi|^{2}+|\nabla \phi|^{4}
\ee
Hence, if $m^{2}\geq -2\Lambda/(n-1)$ then $\phi$ is constant by Corollary \ref{CCOR}. But if $\phi$ is constant and $m^{2}> 0$ then $\phi$ must be indeed zero by equation (\ref{SESF4}). 

Let us assume then that $m^{2}<-2\Lambda/(n-1)$. By Corollary \ref{CCOR} we have 
\be
|\nabla \phi|^{2}\leq -2(\frac{2\Lambda}{(n-1)}+m^{2})\leq \frac{-4\Lambda}{(n-1)}
\ee
which shows the first estimate of (\ref{PHIEST}). Using this estimate together with $m^{2}<-2\Lambda/(n-1)$ we deduce 
\be\label{NPHIEST}
m|\nabla \phi|\leq \frac{\sqrt{8}\sqrt{-\Lambda}}{(n-1)}
\ee
The convenience of this estimate is the following. If two points $p_{0}$ and $p$ are separated by a distance $L$ then
\begin{align}\label{INCRE}
m|\phi(p_{0})| -m|\phi(p)| & \leq |m\phi(p_{0})-m\phi(p)|\\ 
& =\big|\int_{\gamma} m\nabla_{\gamma'}\phi\, ds\big|\\ 
& \leq \frac{\sqrt{8}\sqrt{-\Lambda}}{(n-1)}L
\end{align}
where $\gamma(s)$ is a length minimising geodesic segment joining $p_{0}$ to $p$. Hence, if at a point $p_{0}$ we have
\be\label{IF}
m|\phi(p_{0})|\geq 8\sqrt{-\Lambda}
\ee
then 
\be\label{MOM}
m|\phi(p)|\geq 5\sqrt{-\Lambda}
\ee
at every point $p$ of the ball $B(p_{0},(n-1)/\sqrt{-\Lambda})$ because using (\ref{INCRE}) we would have $m|\phi(p)|\geq 8\sqrt{-\Lambda}-2\sqrt{2}\sqrt{-\Lambda}\geq 5\sqrt{-\Lambda}$. Assume then that (\ref{MOM}) holds on $B(p_{0},(n-1)/\sqrt{-\Lambda})$. Then by (\ref{SESF1}) we would have
\be
Ric^{1}_{f}\geq \bigg(\frac{-23\Lambda}{n-1}\bigg)g=(nH)g
\ee
where the r.h.s is the definition of $H$. But then the radius of the ball, $(n-1)/\sqrt{-\Lambda}$, should be less or equal than $\pi/\sqrt{H}$, in other words we should have
\be
\frac{n-1}{\sqrt{-\Lambda}}\leq \frac{\pi \sqrt{n(n-1)}}{\sqrt{23}\sqrt{-\Lambda}}
\ee
But his equation doesn't hold for any $n\geq 2$. Thus, (\ref{IF}), (hence (\ref{MOM})), cannot hold and we have
\be
m^{2}|\phi|^{2}\leq -64\Lambda
\ee
which is the second estimate of (\ref{PHIEST}).
\end{proof}
So far, Theorem \ref{LPHIEST} provides complete estimates for the scalar field $\phi$. We occupy now ourselves with the Lorentzian geometry, namely with $N$ and $g$. As we show below, gradient estimates for $\ln N$ can be provided in any dimension but pointwise curvature estimate only in spatial dimension three. We start proving estimates for $N$. 
\begin{Theorem}\label{EKGLEDEST} Let $(\Sigma;g,N,\phi)$ be a geodesically complete solution of the static Einstein-ScalarField equations with potential 
$V(\phi)=m^{2}|\phi|^{2}+2\Lambda$, where $\Lambda<0$. Then, the following holds,
\begin{enumerate}
\item[(i)] if $m^{2}\geq -2\Lambda/(n-1)$, then
\be\label{NESTIO}
\frac{|\nabla N|}{N}\leq \sqrt{\frac{-2\Lambda}{n-1}}
\ee
and,
\item[(ii)] if $m^{2}< -2\Lambda/(n-1)$ then
\be\label{NESTIOI}
\frac{|\nabla N|}{N}\leq 64\sqrt{-\Lambda}
\ee
\end{enumerate}
\end{Theorem}  
\begin{proof} During the proof we use $f=-\ln N$. Use (\ref{BOC}) with $\chi=f$, discard a pair of terms and obtain
\be
\frac{1}{2}\Delta_{f}|\nabla f|^{2}\geq \langle \nabla f,\frac{m^{2}\nabla (|\phi|^{2})}{n-1}\rangle +\frac{2\Lambda}{n-1}|\nabla f|^{2}+|\nabla f|^{4}
\ee

If $m^{2}\geq -2\Lambda/(n-1)$ then $\nabla \phi=0$ and the first term in the r.h.s of the previous equation is zero. We can use Corollary  \ref{CCOR} to obtain $|\nabla f|^{2}\leq -2\Lambda/(n-1)$, which is (\ref{NESTIO}). 

Assume now that $m^{2}< -2\Lambda/(n-1)$. We need to bound the first term in the r.h.s of the previous equations. We do this as follows. First write
\begin{align}
|\langle \nabla f, & \frac{m^{2}\nabla (|\phi|^{2})}{n-1}\rangle| = \\ 
& =  |2m^{2}(\phi_{R}\langle \nabla f,\nabla \phi_{R}\rangle +\phi_{I}\langle \nabla f,\nabla \phi_{I}\rangle| \\
\label{TOBOU} & \leq 2m(m|\phi_{R}||\nabla \phi_{R}|+m|\phi_{I}||\nabla \phi_{I}|)|\nabla f| 
\end{align}
Use now Theorem (\ref{LPHIEST}) to bound (\ref{TOBOU}) as
\be
2m(m|\phi_{R}||\nabla \phi_{R}|+m|\phi_{I}||\nabla \phi_{I}|)|\nabla f| \leq \frac{128}{n-1}(-\Lambda)^{3/2}|\nabla f|
\ee
Thus,
\be
\langle \nabla f,\frac{m^{2}\nabla (|\phi|^{2})}{n-1}\rangle\geq -\frac{128}{n-1}(-\Lambda)^{3/2}|\nabla f|
\ee
Hence
\be
\frac{1}{2}\Delta_{f}|\nabla f|^{2}\geq  -\frac{128}{n-1}(-\Lambda)^{3/2}|\nabla f|+\frac{2\Lambda}{n-1}|\nabla f|^{2}+|\nabla f|^{4}
\ee
Making $\psi=|\nabla f|^{2}$ we can write
\be\label{FUNDEQI}
\Delta_{f}\psi\geq a\sqrt{\psi}+b\psi+c\psi^{2}
\ee
where $a=-256(-\Lambda)^{3/2}/(n-1)$, $b=4\Lambda/(n-1)$ and $c=2$. This equation is not the same as (\ref{FUNDEQ}) and Corollary \ref{CCOR} cannot be directly used. However a simple modification of the arguments of Lemma \ref{MAIN} shows that, if (\ref{FUNDEQI}) holds, then 
\be
\psi(p)\leq \max\bigg\{\bigg(\frac{a}{b}\bigg)^{2},-\frac{2b}{c}\bigg\}
\ee
Using this with the values of $a, b$ and $c$ given before we obtain (\ref{NESTIOI}).
\end{proof}
The following theorem proves that, when the spatial dimension is three, the Ricci curvature is bounded by an expression depending only on $\Lambda$. The proof uses some advanced elements of Riemannian geometry. 

\begin{Theorem}\label{RCEST}Let $(\Sigma;g,N,\phi)$ be a geodesically complete solution of the static Einstein-ScalarField equations with potential 
$V(\phi)=m^{2}|\phi|^{2}+2\Lambda$, where $\Lambda<0$ and in spacetime dimension four (i.e. $n=3$). Then,
\be\label{RICEST}
|Ric|\leq \mathcal{R}(|\Lambda|)
\ee
for some non-negative function $\mathcal{R}$.
\end{Theorem}

\begin{proof} In Theorems \ref{LPHIEST} and \ref{EKGLEDEST} we deduced pointwise bounds for $|\nabla f|$, $|\nabla \phi|$ and for $m^{2}|\phi|^{2}$ depending only on $\Lambda$. Therefore, recalling (\ref{SESF1}), the estimate (\ref{RICEST}) would follow granted we can prove a pointwise estimate of $|\nabla\nabla f|$ depending only on $\Lambda$. We prove now that this is possible when $n=3$.

Let $p$ be an arbitrary point in $\Sigma$. Assume that $N(p)=1$. (If $N(p)\neq 1$ then work with the scaled lapse $N/N(p)$. Observe that the system (\ref{SESF1})-(\ref{SESF4})  is invariant under scalings of the lapse. Below we use therefore $f=-\ln N$ and we assume $N(p)=1$). 

To start note that the estimates of Theorem \ref{EKGLEDEST} imply\footnote{Just integrate $\nabla \ln N$ along radial geodesics and used then the bound $|\nabla \ln N|\leq 64\sqrt{-\Lambda}$.} 
\be\label{fBOUND}
|f|(q)\leq K_{0}(\Lambda) 
\ee
for every $q$ in $B_{g}(p,1)$ and where $K_{0}(\Lambda) =64\sqrt{-\Lambda}$. Hence we can write
\be\label{NESTT}
K_{1}(\Lambda)^{-1}\leq N(q)\leq K_{1}(\Lambda),
\ee
for every $q$ in $B_{g}(p,1)$ and where $K_{1}(\Lambda)=e^{K_{0}(\Lambda)}$. As we mentioned earlier, the Theorems  \ref{LPHIEST} and \ref{EKGLEDEST} give us suitable bounds for $|\nabla f|$, $|\nabla \phi|$ and for $m|\phi|$. From such bounds one can write down the coarse estimate
\be\label{KTWPEST}
|\nabla f| + |\nabla \phi| + m|\phi|\leq K_{2}(\Lambda)
\ee
for some $K_{2}(\Lambda)$. This is all what we will need later. We will refer to it a couple of times.  

From now on we will use the metric 
\be
\check{g}:=N^{2}g
\ee
In terms of the variables $(\check{g},N,\phi)$, the static equations (\ref{SESF1})-(\ref{SESF4}) are, 
\begin{align}
\label{SESF1C} & \check{Ric} = 2\nabla f\nabla f +\nabla \phi\circ \nabla \bar{\phi} +\frac{V(\phi)}{2} e^{2f} \check{g},\\
\label{SESF3C} & \check{\Delta} f = \frac{1}{2}V(\phi)e^{f},\\
\label{SESF4C} & \check{\Delta} \phi = \frac{1}{2}\partial V(\phi) e^{f},
\end{align}
Now, use the bounds (\ref{fBOUND}) and (\ref{KTWPEST}) in the formula (\ref{SESF1C}) to deduce that $|\check{Ric}|_{\check{g}}$ is pointwise bounded in $B_{g}(p,1)$, where the bound depends only on $\Lambda$. Thus we have
\be\label{RICCHB}
|\check{Ric}|_{\check{g}}\leq K_{3}(\Lambda)
\ee
As we are working in dimension three, where the Riemann tensor is made out of the Ricci tensor, the bound (\ref{RICCHB}) implies a bound also for the Riemann tensor $\check{Rm}$ on $B_{g}(p,1)$ and thus we have,
\be\label{RMEST}
|\check{Rm}|_{\check{g}}\leq K_{4}(\Lambda)
\ee

Now, it is direct to see from (\ref{NESTT}) that one can find $\check{r}_{1}(\Lambda)$ such that 
\be
B_{\check{g}}(p,\check{r}_{1})\subset B_{g}(p,1/2).
\ee 
Moreover, it is a standard fact in Riemannian geometry that a bound on the Riemann tensor as (\ref{RMEST}) implies that, for some $\check{r}_{2}(\Lambda)<\check{r}_{1}(\Lambda)$, the exponential map
\be
exp:U(p,\check{r}_{2})\rightarrow B_{\check{g}}(p,\check{r}_{2}) 
\ee
is a smooth cover, where in this formula $U(p,\check{r}_{2})$ is the ball of radius $\check{r}_{2}$ in $T_{p}\Sigma$, (endowed with the metric $\check{g}(p)$, namely $U(p,r):=\{v\in T_{p}\Sigma:|v|_{\check{g}(p)}\leq r\}$). 
Provide now $U(p,\check{r}_{2})$ with the pull-back metric $\check{g}^{*}=exp^{*}\check{g}$. The injectivity radius at $p$ of the space $(U(p,\check{r}_{2}),\check{g}^{*})$ is of course equal to $\check{r}_{2}$ and the Riemann tensor of $\check{g}^{*}$ is subject to the same bound (\ref{RMEST}) as $\check{g}$. Therefore, the {\it harmonic radius} of the space $(U(p,\check{r}_{2}), \check{g}^{*})$ at $p$ is controlled from below only by $\Lambda$, (see \cite{MR2243772}, $\S$ Chp. 10.5.2). To us, the only important consequence of this is that one can make standard elliptic analysis on $(U(p,\check{r}_{3}),\check{g}^{*})$ for a suitable $\check{r}_{3}(\Lambda)\leq\check{r}_{2}(\Lambda)$. 
Hence, we can use the bounds (\ref{fBOUND})-(\ref{KTWPEST}) to obtain Schauder interior elliptic estimates from the elliptic system (\ref{SESF3C})-(\ref{SESF4C}), (see \cite{MR2243772}, $\S$ Chp. 10.2). Doing so we get
\be\label{BBB}
|\check{\nabla}\check{\nabla} f|_{\check{g}}(p)\leq K_{5}(\Lambda)
\ee  
Use now the expression,
\begin{align}
\check{\nabla}_{i}\check{\nabla}_{j}f=\nabla_{i} \nabla_{j}f+2\nabla_{j}f\nabla_{i}f -|\nabla f|_{g}^{2}g_{ij}
\end{align}
and the bounds (\ref{BBB}), (\ref{fBOUND}) and (\ref{KTWPEST}), to deduce directly the bound
\be
|\nabla\nabla f|_{g}(p)\leq K_{6}(\Lambda)
\ee
as wished.
\end{proof}

\subsection{Real scalar fields\label{RSF}}
General interesting results can be obtained when $\phi$ is real. The following theorem, gives a simple condition for $V(\phi)$ that forces $\phi$ to be a constant. It gives nice applications that will be illustrated very briefly below.
\begin{Theorem} Let $(\Sigma;g,N,\phi)$ be a geodesically complete solution of the static Einstein-RealScalarField system with potential $V(\phi)$. If $V$ is bounded below and
\be\label{CONDV}
V''(x)+\frac{V(x)}{n-1}\geq 0,
\ee
for all $x$, then $\phi=\phi_{0}$, (constant), and $\phi_{0}$ is a critical point of $V(\phi)$. 
\end{Theorem} 

\begin{proof} Make $f=-\ln N$. Then, using (\ref{BOC}) with $\chi=\phi$ we obtain
\be\label{TOUT}
\frac{1}{2}\Delta_{f}|\nabla \phi|^{2}\geq \big(V''(\phi)+\frac{V(\phi)}{n-1}\big)|\nabla \phi|^{2}+|\nabla \phi|^{4}
\ee
If (\ref{CONDV}) holds an $\Sigma$ is compact then $\nabla \phi=0$ by integrating (\ref{TOUT}) over $\Sigma$. On the other hand if $\Sigma$ is non-compact and (\ref{CONDV}) holds then $\nabla \phi=0$ from Corollary \ref{CCOR}. 

Finally if $\phi=\phi_{0}$ then equation (\ref{SESF4}) shows that $\phi_{0}$ is a critical point of $V(\phi)$.
\end{proof}
To illustrate the relevance of this Theorem let us consider a set of simple and (more or less) natural potentials and let us enumerate, without entering into further discussion, the strong conclusions that can be deduced in each case. 
\begin{enumerate}
\item $V(\phi)=\lambda \phi^{2n}$, $\lambda>0$, $n=1,2,3,\ldots$. In this case (\ref{CONDV}) is verified and therefore any geodesically complete solution must have $\phi=0$. 
\item $V(\phi)=\lambda\cosh \phi$, $\lambda>0$. In this case (\ref{CONDV}) is verified and therefore any geodesically complete solution must have $\phi=0$.
\item $V(\phi)=\lambda e^{\phi}$, $\lambda>0$. In this case (\ref{CONDV}) is verified but there cannot be geodesically complete solutions at all because $V$ has no critical points.
\item $V(\phi)=\lambda \sin \sqrt{(n-1)}\phi$ (a type of Sine-Gordon potential). In this case the l.h.s of (\ref{CONDV}) is identically zero and thus any geodesically complete solution must have $\phi = (-\pi/2+2j\pi)/\sqrt{n-1}, j\in \mathbb{Z}$ (the other critical points make $V$ strictly positive). This example is interesting because it shows that strong conclusions can be obtained even when $V$ is not a non-negative potential.
\item $V(\phi)=\lambda (\phi^{2}-\phi_{0}^{2})^{2}$, $\lambda>0$, (a type of Higgs potential). In this case one can show that if $\phi^{2}_{0}>6(n-1)$ then any geodesically complete solution must have $|\phi|=|\phi_{0}|$. To see this observe that, in this case, (\ref{CONDV}) is equivalent to
\be
12(\phi^{2}-\phi_{0}^{2})+8\phi^{2}_{0}+\frac{(\phi^{2}-\phi_{0}^{2})^{2}}{n-1}\geq 0
\ee
Making $z=\phi^{2}-\phi_{0}^{2}$, the previous equation is equivalent to $12z+8\phi_{0}^{2}+z^{2}/(n-1)\geq 0$ for all $z\geq -\phi_{0}^{2}$. But if $\phi_{0}^{2}\geq 6(n-1)$ then the polynomial $12z+8\phi_{0}^{2}+z^{2}/(n-1)$ is non-negative. 
\end{enumerate}

\bibliographystyle{plain}
\bibliography{Master}

\end{document}